\documentclass[a4paper,11pt]{article}

\usepackage{a4wide}

\usepackage{geometry}                
\usepackage{amssymb, amsmath, amsthm}

\usepackage{tikz}
\usepackage{hyperref}
\usepackage[ruled, vlined]{algorithm2e}

\newtheorem{theorem}{Theorem}[section]
\newtheorem{lemma}[theorem]{Lemma}
\newtheorem{property}[theorem]{Property}
\newtheorem{proposition}[theorem]{Proposition}
\newtheorem{corollary}[theorem]{Corollary}

\newtheorem{definition}[theorem]{Definition}

\newtheorem{remark}[theorem]{Remark}

\newcommand{\mP}{\mathcal{P}}

\newcommand{\rs}{\sigma^{-1}}

\newcommand{\LexDFS}{LDFS\xspace}
\newcommand{\LexBFS}{LBFS\xspace}

\newcommand{\PN}{P(\mathbb{N}^+)}
\newcommand{\PNF}{P_f(\mathbb{N}^+)}

\newcommand{\Ni}{\mathbb{N}^+_i}

\newcommand{\siginv}{\sigma^{-1}}
\newcommand{\sigp}{\gamma}
\newcommand{\siginvp}{\sigp^{-1}}

\newcommand{\pregtls}{\prec_{TBLS}}
\newcommand{\pregls}{\prec_{GLS}}
\newcommand{\pregen}{\prec_{gen}}
\newcommand{\prebfs}{\prec_{BFS}}
\newcommand{\predfs}{\prec_{DFS}}

\newcommand{\premcs}{\prec_{MCS}}
\newcommand{\premns}{\prec_{MNS}}
\newcommand{\prelbfs}{\prec_{\LexBFS}}
\newcommand{\preldfs}{\prec_{\LexDFS}}

\title{ A tie-break model for graph search}
\author{Derek G. Corneil \thanks{Dept. of Computer Science, University of Toronto}
\and 
J\'er\'emie Dusart  \thanks{LIAFA, UMR  7089 CNRS  \& Universit\'e Paris Diderot, F-75205 Paris Cedex 13, France}
\and
Michel Habib$^{~\dag}$
\and
Antoine Mamcarz$^{~\dag}$
\and Fabien de Montgolfier$^{~\dag}$}

\begin{document}

\maketitle

\begin{abstract}
In this paper, we consider the problem of the recognition of various kinds of orderings produced by graph searches. To this aim, we introduce a new framework, the  Tie-Breaking Label Search (TBLS), in order to handle a broad variety of searches. This new model is based on partial orders defined on the label set and it unifies the General Label Search (GLS) formalism of Krueger, Simonet and Berry (2011), and the ``pattern-conditions'' formalism of Corneil and Krueger (2008). It allows us to derive some general properties including new pattern-conditions (yielding memory-efficient certificates) for many usual searches, including BFS, DFS, LBFS and LDFS.  Furthermore, the new model allows easy expression of multi-sweep uses of searches that depend on previous (search) orderings of the graph's vertex set.
\end{abstract}

\textbf{Keywords:}
Graph search model, tie-break mechanisms, BFS, DFS, LBFS, LDFS, multi-sweep algorithms.

\section{Introduction}\label{sec:intro} 

A graph search is a mechanism for systematically visiting the vertices 
of a graph. It has been a fundamental technique in the design of graph algorithms since the early days of computer science.  Many of the early search methods were based on Breadth First Search (BFS) or Depth First Search (DFS) and resulted in efficient algorithms for practical problems such as the distance between two vertices, diameter, connectivity, network flows and the recognition of planar graphs see \cite{CormenLR89}. 

Many variants of these searches have been introduced since, providing elegant and simple solutions to many problems.  For example, Lexicographic Breadth First Search (LBFS) ~\cite{RTL}, and its generalization Maximal Neighbourhood Search (MNS)  \cite{Shier}  were shown to yield simple linear time algorithms for chordal graph recognition.  More recently, Lexicographic Depth First Search (\LexDFS), was introduced in \cite{CK} based on its symmetrical ``pattern-condition'' with \LexBFS.  A few years after its discovery it was shown that LDFS when applied to cocomparability graphs yields simple and efficient algorithms for solving various Hamiltonian Path related problems \cite{CDH, MC, CHK}. 

Some recent applications of graph searches involve a controlled tie-break mechanism in a series of consecutive graph searches, see \cite{Corneil04, DOS09, CDH, DH13}. Examples include the  strongly connected components computation using double-DFS \cite{S81}
 and the series of an arbitrary LBFS followed by two LBFS$^+$s used to recognize unit interval graphs \cite{Corneil04}.
Note that a ``$^+$ search'' breaks ties by choosing (amongst the tied vertices) the vertex that is rightmost with respect to a given ordering
of the vertices.
This motivates a general study of these graph searches equipped with a tie-break mechanism that incorporates such multi-sweep usage of graph searches. This is the goal of the present paper: to define the simplest  framework powerful enough to capture many graph searches either used individually or in a multi-sweep fashion and simple enough to allow general theorems on graph searches.  Building on the General Label Search (GLS) framework from  \cite{BGS11} we not only simplify their model but also unify their model with the ``pattern-conditions'' formalism of \cite{CK}.

This paper is organised as follows. After basic notations and definitions in Section 2, Section~\ref{sectTBLS} introduces the Tie-Breaking Label Search (TBLS) formalism to address graph searches. We then illustrate the TBLS by expressing some classical graph searches in this formalism.  We will also show the relationship between our formalism and the ``pattern-conditions'' of search orderings introduced in  \cite{CK} thereby yielding some new pattern-conditions for various classical searches. In  Section~\ref{compa} we show that the TBLS and GLS models capture the same set of graph searches. We then  propose a unified method for recognizing whether a given ordering of the vertices could have been produced by a specific graph search.  Finally, in Section~\ref{Imp} we present a TBLS implementation framework in a particular case.

\section{Preliminaries and notation}

In this paper, $G=(V,E)$ always (and sometimes implicitly) denotes a graph with $n$ vertices and $m$ edges. All graphs considered here are supposed to be finite.
We identify the vertex-set with $\{1,...,n\}$, allowing us to see a total ordering on $V$ as a permutation on $\{1,...,n\}$.

We define a \emph{graph search} to be an algorithm that visits all the vertices of a graph according to some rules, and a \emph{search ordering} to be the ordering $\sigma$ of the vertices yielded by such an algorithm. The link between these two notions is an overriding theme of this paper.  Vertex $\sigma(i)$ is the $i$th vertex of $\sigma$ and $\siginv(x)\in\{1,...,n\}$ is the position of vertex $x$ in $\sigma$.  A vertex $u$ is the \emph{leftmost} (respectively \emph{rightmost}) vertex with property $X$ in $\sigma$ if there is no vertex $v$ such that $X(v)$ and $v<_\sigma u$ (respectively $u<_\sigma v$).
 Our graphs are assumed to be undirected, but most searches (especially those captured by TBLS) may be performed on directed graphs without any modifications to the algorithm (if $xy$ is an arc then we say that $y$ is a neighbour of $x$ while $x$ is \emph{not} a neighbour of $y$).

The symmetric difference of two sets $A$ and $B$, namely $(A-B) \cup (B-A)$ is denoted by $A \bigtriangleup B$.
Furthermore, $\mathbb{N}^+$ represents the set of integers strictly greater than $0$ and $\mathbb{N}^+_p$ represents the set of integers strictly greater than $0$ and less than $p$. $\PN$ denotes the power-set of $\mathbb{N}^+$ and $\PNF$  denotes the set of all finite subsets of $\mathbb{N}^+$. By $\mathfrak{S}_n $ we denote the set of all permutations of  $\{1,...,n\}$.
For finite $A\in\PN$, let $umin(A)$ be: if  $A=\emptyset$ then $umin(A)=\infty$ else 
$umin(A)=min\{i \mid i\in A\}$; and let $umax(A)$ be: if  $A=\emptyset$ then $umax(A)=0$ else 
$umax(A)=max\{i \mid i\in A\}$.
We always use the notation $<$ for the usual strict (i.e., irreflexive) order between integers, and $\prec$ for a partial strict order between elements from $\PNF$ (or from another set when specified). 
Definitions of most of the searches we will consider appear in  \cite{CK} or  \cite{GOL}. 

\section{TBLS, a  Tie-Breaking Label Search}\label{sectTBLS}

A \emph{graph search} is an iterative process that chooses at each step a vertex of the graph and numbers it (from 1 to $n$). Each vertex is chosen (also said \emph{visited}) exactly once (even if the graph is disconnected).
Let us now define  a general  \emph{Tie-Breaking Label Search} (TBLS). It uses \emph{labels} to decide the next vertex to be visited; $label(v)$  is a subset of $\{1,...,n\}$. A TBLS is defined on:
\begin{enumerate}
\item A graph $G=(V,E)$ on which the search is performed;
\item A strict partial order $\prec$ over the label-set $\PNF$;
\item An ordering $\tau$ of the vertices of $V$ called the \emph{tie-break permutation}.
 \end{enumerate}
 
 The output of $\mbox{TBLS}(G,\prec,\tau)$ is a permutation $\sigma$ of $V$, called a  $TBLS-ordering$ or also the \emph{search ordering}  or \emph{visiting ordering}. Let us say a vertex $v$ is \emph{unnumbered} until $\sigma(i)\leftarrow v$ is performed, and then $i$ is its \emph{visiting date}. Thanks to the following algorithm, $label(v)$ is always the set of visiting dates of the neighbours of $v$ visited before $v$.  
More specifically $label_i(v)$ for a vertex $v$ denotes the label of $v$ at the beginning of step $i$.
 This formalism identifies a search with the orderings it may produce, as in \cite{CK}, while extending
 the formalism of General Label Search (GLS) of \cite{BGS11} by the introduction of a \emph{tie-break} ordering $\tau$, making the result of a search algorithm purely deterministic (no arbitrary decision is taken).  

\vspace{0.5cm}

{\small 
	\begin{algorithm}[ht]

\caption{TBLS($G,\prec,\tau$)\label{gtls}}

\lForEach{$v\in V$}{$label(v)\leftarrow \emptyset$}\
\For{$i\leftarrow 1$ \KwTo $n$}{
$Eligible\gets \{x \in V ~ | ~ x $ unnumbered and $\nexists$  unnumbered $y \in V \text{ such that } label(x) \prec label(y)\}$\;
    Let $v$ be the leftmost vertex of $Eligible$ according to the ordering $\tau$\;
	$\sigma(i)\leftarrow v$\;
	\ForEach{\textup{unnumbered  vertex $w$ adjacent to $v$}}{
		$label(w)\leftarrow label(w)\cup\{i\}$\;
		}
	}
\end{algorithm}
}

\vspace{0.5cm}

\subsection*{Remarks on this formalism: }

\begin{enumerate}
\item
Notice that during a TBLS search vertices are always labelled from $1$ up to $n$.  The original description of \LexBFS generated labels from $n$
down to $1$.  Since a label is always an unordered set rather than a string, as often seen with \LexBFS and \LexDFS, we avoid having
to prepend or append elements to existing labels.
It should also be noticed that the TBLS model does not require the graph to be connected, and therefore in the following we will extend classical graph searches to disconnected graphs.  Since we just need to specify $\prec$ to describe a particular search no implementation details 
have to be discussed in the specification of the search.  Finally, by requiring a tie-breaking permutation $\tau$ we  have a 
mechanism for choosing a specific vertex from $Eligible$.  Many existing recognition algorithms such as the unit interval recognition algorithm
in \cite{Corneil04} use a series of \LexBFS sweeps where ties are broken by choosing the rightmost eligible vertex in the previous \LexBFS
search; to accomplish this in the TBLS formalism $\tau$ is set to be the reverse of the previous \LexBFS ordering.
 
\item
Note that all elements of the set $Eligible$  have a label set which is maximal with respect to some finite partial order,
 since any finite partial order has at least one maximal element; therefore $Eligible$ cannot be empty.  
In the context of \LexBFS the $Eligible$ set  is  often called a \emph{slice}.
The reader should be aware that we make no claims on the complexity of computing 
the strict partial order $\prec$ over the label-set $\PNF$; unfortunately it could be NP-hard.
%


\end{enumerate}

Let us first present  an easy  characterization property of TBLS search that will be used throughout the paper and which is a direct translation of the algorithm.  First  we define $N_\sigma(u,v)$ to be
the set of visiting dates of neighbours of $u$ that occur before $v$ in $\sigma$;  formally $N_\sigma(u,v) = \{i \mid \sigma(i) \in N(u) $ and $ \sigma(i)<_{\sigma} v\}$.

\begin{property}\label{metaordering}
	Let $S$ be a TBLS search with partial order relation $\prec_S$.

	An ordering $\sigma$ of the vertices of a graph G is an S-ordering if and only if for every $x,~y \in V$, if $x <_\sigma y$, then $N_\sigma(x,x) \not\prec_S N_\sigma(y,x)$.
\end{property}

\begin{proof}
The forward direction follows directly from the definition.

For the backward direction, assume that for every $x,~y \in V$, $x <_\sigma y$, $N_\sigma(x,x) \not\prec_S N_\sigma(y,x)$ but $\sigma$ is not an $S$-ordering.  Let $\sigp=\mbox{TBLS}(G,\prec_S,\sigma)$. Since $\sigma$ is not an $S$-ordering, we know that $\sigma \neq \sigp$. Now let $i$ be the first index such that $\siginv(i) \neq \siginvp(i)$. Let $x=\siginv(i)$ and $y=\siginvp(i)$. Since $x <_\sigma y$ but $\mbox{TBLS}(G,\prec_S,\sigma)$ did not choose $x$, we know that $x$ did not have a maximal label at step $i$. Therefore there must exist $z$ such that $x <_\sigma z$, and $label_i(x) \prec_S label_i(z)$. But since $label_i(x)=N_\sigma(x,x)$ and $label_i(z)=N_\sigma(z,x)$ we now have a pair of vertices $x$, $z$ that contradicts the assumption that for all $x,~y \in V$, $x <_\sigma y$, $N_\sigma(x,x) \not\prec_S N_\sigma(y,x)$.

\end{proof}

With this formalism, in order to specify a particular search we just need to specify $\prec$, 
the partial order relation on the label sets for that search.  
As a consequence we can transmit relationships between partial orders to their associated graph searches.

There are two natural ways of saying that a search $S$ \textbf{is a} search $S'$ (for instance, that \LexBFS \textbf{is a} BFS): either the $\prec$ ordering used by $S$ is a refinement of that of $S'$; 
or any search ordering $\sigma$ output by an execution of $S$ could also have been output by an execution of $S'$. 
In fact it can easily be shown that both formulations are equivalent, as stated in Theorem \ref{metaexten}.

\begin{definition}
 	For two TBLS searches $S$, $S'$, we say that $S'$ is an extension of $S$ (denoted by $S \ll S'$) if and only if every $S'$-ordering $\sigma$ also is an $S$-ordering.
\end{definition}

 The statement and proof of the next lemma follows the
work of \cite{BGS11} where there are similar results for the GLS formalism.

\begin{lemma}[see \cite{BGS11}]\label{K}
	For any $TBLS$ S, any integer $p\geq 1$ and any sets $A$ and $B$ of $P(\mathbb{N}^+_p)$, if $A \not\prec_S B$ then 
there exists a  graph $G$ and an $S$-ordering $\sigma$ such that in the $(p-1)$st step the label of the $(p-1)$st vertex is $A$ and the
label of the $p$th vertex  is $B$ (i.e., 
$label_{p-1}(\siginv(p-1))=A$ and $label_{p-1}(\siginv(p))=B$).
\end{lemma}

\begin{proof}
	Let $G=(V,E)$, with $V=\{ z_1,...,z_p\}$ and $E = \{z_iz_k \mid 1 \leq k <i \leq p-2$ and if $A \cap \Ni \prec_S B \cap \Ni$ then $k \in B$ else $k \in A\} \cup \{ z_{p-1}z_k | k \in A\} \cup  \{ z_{p}z_k | k \in B\}$. Let $\sigma = (z_1,\dots,z_p)$. By the definitions of $E$ and $\sigma$, for any integers $i,j$ such that $1 \leq i \leq j \leq p$, $N_\sigma(\siginv(j),\siginv(i))=A \cap \Ni$ or $N_\sigma(\siginv(j),\siginv(i)=B \cap \Ni$ with $N_\sigma(\siginv(i), \siginv(i)) \not\prec_S N_\sigma(\siginv(j), \siginv(i))$. By Property \ref{metaordering}, $\sigma$ is an S ordering and we have  $N_\sigma(\siginv(p-1),\siginv(p-1))=A$ and $N_\sigma(\siginv(p),\siginv(p-1))=B$.
\end{proof}

\begin{definition} For two partial orders $\prec_P, \prec_Q$ on the same ground set $X$, we say that $\prec_P$ is an extension of  $\prec_Q$ if  $\forall x,y \in X$, $x\prec_Q y$ implies   $x\prec_P y$. \end{definition}

\begin{theorem}\label{metaexten}
	Let $S$, $S'$ be two TBLS. $S'$ is an extension of $S$ if and only if $\prec_{S'}$ is an extension of $\prec_S$.   
\end{theorem}

\begin{proof}

		For the forward direction, assume for contradiction that $S'$ is an extension of $S$ but $\prec_{S'}$ is not an extension of $\prec_{S}$. Therefore there exists $A$, $B$ such that $A \prec_S B$ and $A \not\prec_{S'} B$. Now using Lemma \ref{K} there exists a graph $G$ and an $S'$-ordering $\sigma$ such that $label_{p-1}(\siginv(p-1))=A$ and $label_{p-1}(\siginv(p))=B$. Since $A \prec_S B$ using Property \ref{metaordering}, we deduce that $\sigma$ is not an $S$-ordering which contradicts the fact that $S'$ is an extension of $S$.

	Suppose that $\sigma$ is an $S'$-ordering. Therefore using Property \ref{metaordering} we know that for every $x,~y \in V$, $x <_\sigma y$, we have $N_\sigma(x,x) \not\prec_{S'} N_\sigma(y,x)$. Since $\prec_{S'}$ is an extension of $\prec_S$, we deduce that $x,~y \in V$,  such that for every $x <_\sigma y$, we have $N_\sigma(x,x) \not\prec_{S} N_\sigma(y,x)$. Now using Property \ref{metaordering}, we deduce that $\sigma$ is an $S$-ordering.

\end{proof}

The choice of permutation $\tau$ is useful in some situations described below; otherwise, we consider the orderings output by an arbitrary choice of $\tau$ thanks to the following definition:

\begin{definition}\label{defordre}
Let $\prec$ be some ordering over $\PNF$.
Then $\sigma$ is a TBLS ordering for $G$ and $\prec$ if there exists $\tau$ such that $\sigma=\mbox{TBLS}(G,\prec,\tau)$.
\end{definition}

   Before giving some examples
of appropriate $\prec$ for well known searches, let us start with a kind of fixed point theorem and some general features of the TBLS formalism.

\begin{theorem}\label{thm:fixpoint}
Let $G$ be a graph; $\prec$ a search rule; and $\sigma$ an ordering of the vertices of $G$.
Then there exists $\tau$ such that $\sigma=\mbox{TBLS}(G,\prec,\tau)$
if and only if 
$\sigma=TBLS(G,\prec,\sigma)$.

\end{theorem}

\begin{proof}
One direction is obvious. For the other direction, assume that $\sigma=\mbox{TBLS}(G, \prec, \tau)$ for some $\tau$, and consider $\sigma'=\mbox{TBLS}(G,\prec,\sigma)$. Assume, by contradiction, that $\sigma \neq \sigma'$, and consider $i$, the index of the first difference between $\sigma$ and $\sigma'$. Let $Eligible_i^{\sigma}$ be the set of eligible vertices at step $i$ of the algorithm that generated $\sigma$, and let $Eligible_i^{\sigma'}$ be the set of eligible vertices at step $i$ of the algorithm that generated $\sigma'$. 
Since $\sigma$ and $\sigma'$ are equal until index $i$, $Eligible_i^{\sigma}=Eligible_i^{\sigma'}$.
By the definition of TBLS, $\sigma(i)$ is the first vertex of $Eligible_i^{\sigma}$ according to $\tau$.  Finally, since the first vertex of this set, according to $\sigma$, is $\sigma(i)$, the tie-break rule chose it and so $\sigma(i)=\sigma'(i)$, a contradiction.
\end{proof}

This easy result (which is a direct consequence of our model equipped with a built-in tie-break process)
is in fact a very powerful theoretical tool to show that some ordering is not a TBLS ordering, and we will use it several times in the proofs in the next sections, as for example in Theorem \ref{genericsearch}.

As a first conclusion, the TBLS model is a pure mathematical  abstraction of graph search algorithms via partial orders but with no data structures involved. Moreover,
if we have a characterization of the total orderings produced by a given TBLS (as for example the usual search  characterizations of section 4) then we can get rid of the implementation itself which can be parallel or sequential. In the next sections we will exhibit some easy consequences of this model. But before, we 
demonstrate its expressive power, in particular  to deal with multi-sweep algorithms (i.e., algorithms written as a series of successive graph searches).

To this aim, let us consider the sequence $\{\sigma_i \}_{i \in \mathbb{N}}$ of total orderings of the vertices, satisfying the following recursive equations:

(i) $\sigma_0$ is an arbitrary total ordering of the vertices

(ii) $\sigma_i =\mbox{TBLS}(G,\prec,\sigma^r_{i-1})$ where $\sigma^r$ denotes the reverse ordering of $\sigma$.
\vspace{0,5cm}

It was proven when $\prec$ is the partial order associated to the LBFS search, as described in the next section:

(i) in \cite{Corneil04}, that if $G$ is a unit interval graph then $\sigma_{3}$ is a unit interval ordering\footnote{an ordering $\tau$ of $V$ such that for all $x <_{\tau} y <_{\tau} z$ with $xz \in E$, we have $xy, yz \in E$.}. 

(ii) in \cite{DH13} that if $G$ is a cocomparability graph then $\sigma_{|V|}$ is a cocomp ordering\footnote{an ordering $\tau$ of $V$ such that for all $x <_{\tau} y <_{\tau} z$ with $xz \in E$, we have at least one of $xy, yz \in E$.}.

\section{Characterizing classical searches using TBLS}\label{charac}

In this section we show how various classical searches (see \cite{CK} for the definitions of these various searches) may be expressed in the TBLS formalism.  In each case we will state an
appropriate $\prec$ order and where applicable, we will 
establish various characterizations of the search including the ``pattern-condition'' presented in \cite{CK}.  In many cases we will exhibit new
vertex ordering characterizations.


\begin{definition}\label{lnrndef}
For every vertex $x$, let $ln(x)$ be the leftmost (in $\sigma$) left neighbour of $x$, and let $rn(x)$ be the rightmost (in $\sigma$) right neighbour of $x$. In both cases, if $x$ has no left (respectively right) neighbour, then $ln(x)$ (respectively $rn(x)$)=$-1$.
\end{definition}


\subsection{Generic Search}
A \emph{Generic Search} as described by Tarjan \cite{TAR} is any search that wherever possible visits neighbours of already visited vertices (this corresponds to the usual notion of graph search).

We now give an alternative proof based on our formalism of the characterization of a generic search ordering, called a GEN-ordering throughout
the rest of the paper.


\begin{theorem}[see \cite{CK}]\label{genericsearch}
We define $A\pregen B$ if and only if $A=\emptyset$ and $B\ne \emptyset$ and let $\sigma$ be a permutation of $V$.
The following conditions are equivalent:\\[-.6cm]

\begin{enumerate}
	\item Vertex ordering $\sigma$ is a GEN-ordering of $V$ (i.e., a TBLS using $\pregen$). 
	\item For every triple of vertices $a,~b,~c$ such that $ a <_\sigma b <_\sigma c$, and $a \in N(c)-N(b)$ there exists $d \in N(b)$ such that $d <_\sigma b$.
     \item For every $x\in V$, for every $y\in V$ such that $x<_{\sigma}y <_{\sigma}rn(x)$, we have $ln(y) \neq -1$.
\end{enumerate}
\end{theorem}

\begin{proof}
	Suppose that $\sigma$ is a GEN-ordering. 
Using Property \ref{metaordering} on $\sigma$, we know that:\\
 $\sigma$ is a GEN-ordering\\
		$\iff$ for every $x,~y \in V$ such that  $x <_\sigma y$, we have $N_\sigma(x,x) \not\pregen N_\sigma(y,x)$\\
		$\iff$ for every $x,~y \in V$ such that  $x <_\sigma y$,  we have $N_\sigma(y,x) = \emptyset$ or $N_\sigma(x,x) \neq \emptyset$\\
		$\iff$ for every $x,~y \in V$ such that   $x <_\sigma y$, we have $N_\sigma(y,x) \neq \emptyset \Rightarrow N_\sigma(x,x) \neq \emptyset$\\
		$\iff$ for every triple of vertices $a,~b,~c$ such that $ a <_\sigma b <_\sigma c$, $a \in N(c)-N(b)$, there exists $d \in N(b)$ such that $d <_\sigma b$.\\

Therefore we proved the equivalence between 1 and 2. Let us consider 3, which is a reformulation of 2.  The fact that $1 \Rightarrow 3$, is obvious. To prove the converse, we can use Theorem \ref{thm:fixpoint}. Suppose that there exists $\sigma$ satisfying  3 but not 1. Let $\sigma'=\mbox{TBLS}(G,\pregen,\sigma) \neq \sigma$ and $i$ be the leftmost index where they differ ($z=\sigma'(i)\neq y= \sigma(i)$). This means  that  $l_i(y)=\emptyset$ and there exists some $x \in l_i(z)$. But this implies with condition 3
since $x<_{\sigma}y <_{\sigma}z \leq_{\sigma}rn(x)$ that $ln(y) \neq -1$ contradicting $l_i(y)=\emptyset$.
\end{proof}

\begin{remark}
Theorem \ref{thm:fixpoint} can be used also in the following proofs of this section, but we omit them to avoid tedious reading.
\end{remark}

\subsection{BFS (Breadth-First Search)}\label{sec:bfs}

We now focus on BFS. Here, we will follow the definition of BFS given in \cite{CK}, that is a graph search in which the vertices that are eligible are managed with a queue. Note that this differs for example from the definition given in \cite{CormenLR89}, where BFS stands for what we call layered search. Our notion of BFS is the most common implementation of a layered search.



\begin{theorem}\label{th:bfs}
We define $A\prebfs B$ if and only if $umin(A)>umin(B)$. Let $\sigma$ 
 be a permutation of $V$.
The following conditions are equivalent:\\[-.6cm]

\begin{enumerate}
		\item Vertex ordering $\sigma$ is a BFS-ordering (i.e., a TBLS using $\prebfs$). 
		\item For every  triple $a,b,c \in V$ such that $a<_\sigma b <_\sigma c$, $a \in N(c) - N(b)$, there exists d such that $d \in N(b)$ and $d <_\sigma a$.
		\item For every  triple $a,b,c \in V$ such that $a<_\sigma b <_\sigma c$ and $a$ is the leftmost vertex of $ N(b) \cup N(c)$
		in $\sigma$, we have $a \in N(b)$.
		
	\end{enumerate}
\end{theorem}

\begin{proof}
	The equivalence between condition (1) and condition (2) has been proved in \cite{CK}. We now prove that condition (1) is equivalent to condition (3).
	Suppose that $\sigma$ is a BFS-ordering. 
Using Property \ref{metaordering} on $\sigma$, we know that:\\ $\sigma$ is a BFS-ordering\\
		$\iff$ for every $x,~y \in V$, $x <_\sigma y$, $N_\sigma(x,x) \not\prebfs N_\sigma(y,x)$\\
		$\iff$ for every $x,~y \in V$, $x <_\sigma y$, $umin(N_\sigma(x,x)) \not> umin(N_\sigma(y,x))$ \\
		$\iff$ for every $x,~y \in V$, $x <_\sigma y$, $umin(N_\sigma(x,x)) \leq umin(N_\sigma(y,x))$\\
		$\iff$ for every triple of vertices $a,~b,~c$ such that $a<_\sigma b <_\sigma c$, and $a$ is the leftmost vertex of $ N(b) \cup N(c)$, we have $a \in N(b)$.\\

\end{proof}

\subsection{DFS (Depth First Search)}\label{sec:dfs}

We now turn our attention to Depth First Search.

\begin{theorem}\label{th:dfs}
We define $A\predfs B$ if and only if $umax(A)<umax(B)$.
Let $\sigma$  be a permutation of $V$.
The following conditions are equivalent:\\[-.6cm]

\begin{enumerate}
		\item Vertex ordering $\sigma$ is a DFS-ordering (i.e., a TBLS using $\predfs$). 
		\item For every triple of vertices $a,~b,~c$ such that $ a <_\sigma b <_\sigma c$, $a \in N(c)-N(b)$ there exists $d \in N(b)$ such that 
		$a <_\sigma d <_\sigma b$. 
		\item For every triple of vertices $a,~b,~c$ such that $ a <_\sigma b <_\sigma c$, and $a$ is the rightmost vertex of $N(b) \cup N(c)$ to the left of $b$ 
		in $\sigma$, we have $a \in N(b)$.
	\end{enumerate}
\end{theorem}

\begin{proof}
	The equivalence between condition (1) and (2) has been proved in \cite{CK}.	Let us show the equivalence between (1) and (3).
	Suppose that $\sigma$ is a DFS-ordering. 
Using Property \ref{metaordering} on $\sigma$, we know that:\\ $\sigma$ is a DFS-ordering\\
		$\iff$ for every $x,~y \in V$ such that $x <_\sigma y$, we have $N_\sigma(x,x) \not\predfs N_\sigma(y,x)$\\
		$\iff$ for every $x,~y \in V$ such that $x <_\sigma y$, we have $umax(N_\sigma(x,x)) \not< umax(N_\sigma(y,x))$\\
		$\iff$ for every $x,~y \in V$ such that $x <_\sigma y$, we have $umax(N_\sigma(y,x)) \leq umax(N_\sigma(x,x))$\\
		$\iff$ for every triple of vertices $a,~b,~c$ such that $a<_\sigma b <_\sigma c$ and $a$ is the rightmost vertex of $ N(b) \cup N(c)$ to the left of $b$ in $\sigma$, we have $a \in N(b)$.\\

\end{proof}

\subsection{\LexBFS (Lexicographic Breadth First Search)}\label{sec:lbfs}

\LexBFS  was first introduced in \cite{RTL} to recognize chordal graphs.  Since then many new applications of \LexBFS have been
presented ranging from recognizing various  families of graphs to finding vertices with high eccentricity or to finding the modular 
decomposition of a given graph,
see \cite{HMPV00,SURV,DOS09,Ted}.  

 \begin{theorem}\label{th:lbfs}
 We define $A \prelbfs B$ if and only if $umin(B-A)<umin(A-B)$. 
Let $\sigma$   be a permutation of $V$.
The following conditions are equivalent:\\[-.6cm]

\begin{enumerate}
		\item Vertex ordering $\sigma$ is a \LexBFS-ordering (i.e., a TBLS using $\prelbfs$)
		\item For every triple $a,b,c \in V$ such that $a<_\sigma b <_\sigma c$, $a \in N(c)-N(b)$,  there exists $d <_\sigma a$, $d\in N(b)-N(c)$.
		\item For every triple $a,b,c \in V$ such that $a<_\sigma b <_\sigma c$ and $a$ is the leftmost vertex of $N(b) \bigtriangleup N(c)$ to the left of $b$
		in $\sigma$, then a $\in N(b)-N(c)$.
	\end{enumerate}
 \end{theorem}
 
 \begin{proof}
	 The equivalence between (1) and (2) is well known, see \cite{RTL,GOL,BD97}. We now prove the equivalence between (1) and (3).
	 Suppose that $\sigma$ is a \LexBFS-ordering.
 Using Property \ref{metaordering} on $\sigma$, we know that:\\ $\sigma$ is a \LexBFS-ordering\\
		$\iff$ for every $x,~y \in V$, $x <_\sigma y$, we have $N_\sigma(x,x) \not\prelbfs N_\sigma(y,x)$\\
		$\iff$ for every $x,~y \in V$, $x <_\sigma y$, we have $umin(N_\sigma(y,x) - N_\sigma(x,x)) \not< umin(N_\sigma(x,x) -N_\sigma(y,x))$\\
		$\iff$ for every $x,~y \in V$, $x <_\sigma y$, we have $umin(N_\sigma(x,x)-N_\sigma(y,x)) \leq umin(N_\sigma(y,x)-N_\sigma(x,x))$\\
		$\iff$ for every triple of vertices $a,~b,~c$ such that $a<_\sigma b <_\sigma c$ and $a$ is the leftmost vertex of $ N(b) \bigtriangleup N(c)$ to the left of $b$ in $\sigma$, we have $a \in N(b)-N(c)$.\\
\end{proof}

\subsection{\LexDFS (Lexicographic Depth First Search)}\label{sec:ldfs}
Lexicographic Depth First Search (\LexDFS) was introduced in \cite{CK}.  

\begin{theorem}\label{th:ldfs}
We define $A \preldfs B$ if and only if $umax(A-B)<umax(B-A)$.
Let $\sigma$  be a permutation of $V$.
The following conditions are equivalent:\\[-.6cm]

\begin{enumerate}
		\item Vertex ordering $\sigma$ is a \LexDFS-ordering (i.e., a TBLS using $\preldfs$)
		\item For every triple $a,b,c \in V$ such that $a<_\sigma b <_\sigma c$, $a \in N(c)-N(b)$,  there exists $a <_\sigma d <_\sigma b$, $d\in N(b)-N(c)$.
		\item For every triple $a,b,c \in V$ such that $a<_\sigma b <_\sigma c$ and $a$ is the rightmost vertex in $N(b) \bigtriangleup N(c)$ to the left of $b$
		in $\sigma$,   $a \in N(b)-N(c)$.
	\end{enumerate}
\end{theorem}
 
 \begin{proof}
 The equivalence between (1) and (2) is well known, see \cite{CK}. We now prove the equivalence between (1) and (3).
 	Suppose that $\sigma$ is a \LexDFS-ordering. 
Using Property \ref{metaordering} on $\sigma$, we know that:\\ $\sigma$ is a \LexDFS-ordering\\
		$\iff$ for every $x,~y \in V$ such that $x <_\sigma y$, we have $N_\sigma(x,x) \not\preldfs N_\sigma(y,x)$\\
		$\iff$ for every $x,~y \in V$ such that $x <_\sigma y$, we  have $umax(N_\sigma(x,x)-N_\sigma(y,x)) \not< umax(N_\sigma(y,x) - N_\sigma(x,x))$\\
		$\iff$ for every $x,~y \in V$ such that  $x <_\sigma y$, we have $umax(N_\sigma(y,x)-N_\sigma(x,x)) \leq umax(N_\sigma(x,x)-N_\sigma(y,x))$\\
		$\iff$ for every triple of vertices $a,~b,~c$ such that $a<_\sigma b <_\sigma c$ and $a$ is the rightmost vertex of $ N(b) \bigtriangleup N(c)$ to the left of $b$ in $\sigma$, we have $a \in N(b)-N(c)$.\\
\end{proof}

The symmetry between BFS and DFS (respectively  \LexBFS and \LexDFS) becomes clear when using the TBLS ordering formalism. This symmetry was also clear using the pattern-conditions as introduced in \cite{CK}, and in fact lead to the discovery of \LexDFS.

To finish with classical searches, we notice that \textbf{Maximum Cardinality Search (MCS)} as introduced in \cite{TY84}, can easily be defined using the partial order:
$A \premcs B$ if and only if $|A| < |B|$.
Similarly \textbf{MNS  (Maximal Neighbourhood Search)} as introduced in ~\cite{Shier} for chordal graph recognition,  is a search such that $A\premns B$ if and only if $A \subsetneq B$, i.e., it uses the strict inclusion partial order between subsets.

To conclude this section we use  Theorem  \ref{metaexten} to easily rediscover the relationships amongst various graph searches as noted in \cite{CK} and \cite{BGS11}.

\begin{theorem}\label{Xhier}
The partial order of the relation extension between  classical searches is described in Figure 1.
\end{theorem}
\begin{proof}
	To show that a search extends another one we will use Theorem \ref{metaexten}.

	Let us show that $\prebfs$ (respectively $\predfs$, $\premns$) is an extension of $\pregen$. Let $A \pregen B$. By definition we have $A = \emptyset$ and $B \neq \emptyset$. As a consequence we have $umin(B) < umin(A)$ and thus $A \prebfs B$ (respectively  $umax(A) < umax(B)$ implying $A \predfs B$, and $A \subsetneq B$ implying $A \premns B$).

	We now show that $\prelbfs$ is an extension of $\prebfs$. Let $A \prebfs B$. By definition, we have $umin(B) < umin (A)$. As a consequence, $umin(B-A) < umin (A-B)$, implying $A \prelbfs B$.

	To see that  $\preldfs$ is an extension of $\predfs$, first suppose that  $A \predfs B$. Therefore $umax(A) < umax (B)$ and as a consequence  $umax(A-B) < umax (B-A)$ thereby implying $A \preldfs B$.

Similarly  $\prelbfs$ (respectively $\preldfs$, $\premcs$) is an extension of $\premns$. Let $A \premns B$;  by definition we have $A \subsetneq B$. As a consequence, $umin(B-A) < umin (A-B)$ (respectively $umax(A-B) < umax (B-A)$ and  $|A| < |B|$). So $A \prelbfs B$ (respectively $A \preldfs B$ and $A \premcs B$).

\end{proof}

\vspace{0.5cm}

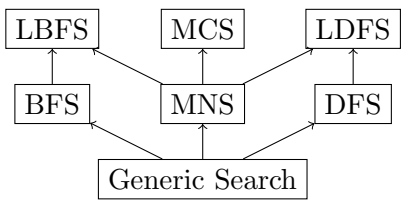
\begin{figure}[ht]\label{fig:panorama}\vspace{-.4cm}
\begin{center}
\begin{tikzpicture}
\node[draw](B) at (3,2) {\LexBFS};
\node[draw](C) at (7,2) {\LexDFS};
\node[draw](D) at (5,2) {MCS};
\node[draw](E) at (5,1) {MNS};
\node[draw](G) at (3,1) {BFS};
\node[draw](H) at (7,1) {DFS};
\node[draw](J) at (5,0) {Generic Search};

\draw[->] (J)--(E);
\draw[->] (J)--(G);
\draw[->] (J)--(H);
\draw[->] (E)--(B);
\draw[->] (E)--(D);
\draw[->] (E)--(C);
\draw[->] (H)--(C);

\draw[->] (G)--(B);

\end{tikzpicture}

\end{center}\vspace{-.4cm}

\vspace{0.5cm}

\caption{Summary of the hereditary relationships proved in Theorem \ref{Xhier}. An arc from Search $S$ to Search $S'$ means that $S'$ extends $S$.}
\end{figure}

\subsection{Limitations of the TBLS model} 

To finish, let us remark that there exists at least one known search that does not fit into the TBLS model.   
In the following, recall that $label_i(v)$ for a vertex $v$ denotes the label of $v$ at the beginning of step $i$ of Algorithm 1. 
\emph{Layered Search} starts at a vertex s, and ensures that if $dist(s,x)<dist(s,y)$ then $\sigma(x)<\sigma(y)$. In other words it respects the layers (vertices at the same distance from the start vertex $s$).
We now show that this search is not a TBLS by considering the graph $G$ in Figure \ref{fig:limit}.  Assume that we have started the Layered Search with $x_1,~x_2,~x_3,~x_4$ and so $label_5(x_5)=\{3\}$ and $label_5(x_6)=\{4\}$. In a Layered Search, both $x_5$ and $x_6$ must be eligible at step 5.  Thus we must have neither $\{3\}\prec \{4\}$  nor $\{4\}\prec \{3\}$; they are incomparable labels.
But now consider graph $H$ in Figure \ref{fig:limit} and assume that again we have started the search with $v_1,~v_2,~v_3,~v_4$. So we have $label_5(v_5)=\{3\}$ and $label_5(v_6)=\{4\}$. But in this graph we have to visit $v_5$ before $v_6$. Therefore we must have  $\{3\} \prec \{4\}$.
As a conclusion,  no partial ordering of the labels can capture all Layered Search orderings  and so this search cannot be written in our formalism. The same seems true for Min-LexBFS as defined in \cite{Meister05} and Right Most Neighbour as  used in \cite{CDH}.

\begin{figure}[ht]
\begin{center}
\begin{tikzpicture}

\coordinate(X1) at (0,0.5);
\coordinate(X2) at (1,0);
\coordinate(X3) at (1,0.5);
\coordinate(X4) at (1,1);
\coordinate(X5) at (2,0.5);
\coordinate(X6) at (2,1);

\draw (X1) node[above left] {$x_1$} node{$\bullet$};
\draw (X2) node[above] {$x_2$} node{$\bullet$};
\draw (X3) node[above] {$x_3$} node{$\bullet$};
\draw (X4) node[above] {$x_4$} node{$\bullet$};
\draw (X5) node[above] {$x_5$} node{$\bullet$};
\draw (X6) node[above] {$x_6$} node{$\bullet$};

\coordinate(V1) at (4,0.5);
\coordinate(V2) at (5,1);
\coordinate(V3) at (5,0);
\coordinate(V4) at (6,1);
\coordinate(V5) at (6,0);
\coordinate(V6) at (7,1);

\draw(V1) node[above left]{$v_1$} node{$\bullet$};
\draw(V2) node[above]{$v_2$} node{$\bullet$};
\draw(V3) node[above]{$v_3$} node{$\bullet$};
\draw(V4) node[above]{$v_4$} node{$\bullet$};
\draw(V5) node[above]{$v_5$} node{$\bullet$};
\draw(V6) node[above]{$v_6$} node{$\bullet$};

\draw (X6)--(X4)--(X1)--(X2);
\draw (X1)--(X3)--(X5);

\draw (V1)--(V2)--(V4)--(V6);
\draw (V1)--(V3)--(V5);
\end{tikzpicture}
\end{center}
\caption{Graph $G$ on the left and $H$ on the right.\label{fig:limit}}
\end{figure}
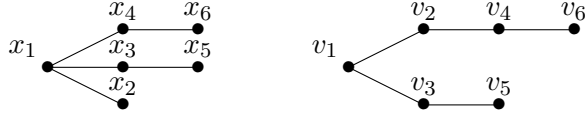

\section{The relationship between GLS and TBLS}
\label{compa}


We are now interested in determining the relationship between TBLS and GLS. First let us recall  GLS from \cite{BGS11}. It depends on a \emph{labeling structure} which consists of four elements:
\begin{itemize}  
\item a set of labels $L$;
\item a strict order  $\pregls$ over the label-set $L$; 
\item an initial label $l_0$;
\item an  UPLAB function $L\times \mathbb{N}^+ \to L$.
\end{itemize}

The GLS algorithm then takes as input a graph $G=(V,E)$ (over which the search is performed) as well as a labeling structure.

The computational power of the UPLAB function is unbounded, even though it must be deterministic, and the label set $L$ may be any set.
In contrast, TBLS uses a fixed initial label $\emptyset$, a fixed label set $\PNF$, and a fixed simple updating function. Despite these restrictions, it is, however equivalent to $GLS$ in the sense of Theorem \ref{equivalence}.

\vspace{0.5cm}

{\small
	\begin{algorithm}[H] 
\caption{GLS($G,\{L,\pregls,l_0,UPLAB\}$)}

\lForEach{$v\in V$}{$l(v)\leftarrow l_0$}\
\For{$i\leftarrow 1$ \KwTo $n$}{
        Let $Eligible$ be the set of eligible vertices, i.e., those unnumbered vertices $v$ with $l(v)$ maximal with respect to $\pregls$\;
        Let $v$ be some vertex from $Eligible$\;
	$\sigma(i)\leftarrow v$\;
	\ForEach{\textup{unnumbered  vertex $w$ adjacent to $v$}}{
		$l(w)\leftarrow UPLAB(l(w),i)$\;
		}
	}
\end{algorithm}
}
\vspace{0.5cm}

We now prove that for each GLS, there is a $\pregtls$ producing the same orderings, and conversely. First we need some notation.


At each iteration $i$ of $\mbox{TBLS}(G,\prec_{TBLS}, \sigma)$, let $l_{TBLS,i}(v)$ be the label assigned to every unnumbered vertex  $v$ by $\mbox{TBLS}(G,\prec_{TBLS}, \sigma)$, i.e., the label that will be used to choose the $i$th vertex. Similarly, let $l_{GLS,i}(v)$ be the label assigned to every unnumbered vertex $v$ by GLS($G, \{L,\pregls,l_0,UPLAB\}$), i.e., the label that will be used to choose the $i$th vertex.
Given a graph $G=(V,E)$, and an ordering $\sigma$ of $V$, let us define $I_k^{\sigma}(v)=N(v) \cap \{\sigma(1), \dots \sigma(k)\}$ to be the neighbours of $v$ visited at step $k$ or before. Let us define $p_k^j$ to be the $j$-th element of $I_k^{\sigma}(v)$ sorted in increasing visiting ordering.

\begin{proposition}[\cite{BGS11}]\label{GLSlabelOK}
Let $S$ be a labeling structure, $G=(V,E)$ a graph. At iteration $i$ of GLS($G,S$) computing an ordering $\sigma$, for every unnumbered vertex $v$:  \\$l_{GLS,i}(v)= UPLAB (\dots UPLAB (l_o, p_1), \dots, p_{k})$  where $(p_1, \dots, p_k)$ is the sequences of
numbers in $N_{\sigma}(v, \rs(i))$ in increasing order.
\end{proposition}

\begin{proposition}\label{TBLSlabelOK}
Let $G=(V,E)$ be a graph, $v\in V$, and   $\sigma$  the ordering produced by TBLS($G,\prec, \tau$). At iteration $i$ of $\mbox{TBLS}(G,\prec, \tau$), for every unnumbered vertex $v$:  \\$l_{TBLS,i}(v)=I_{i-1}^{\sigma}(v)$.
\end{proposition}
\begin{proof}

 The proof goes by induction. At the first step of the algorithm, every vertex has $\varnothing$ as its label, and has no previously visited neighbour.
 
 Assume that at iteration $i$, every unnumbered vertex $x$ has label $l_{TBLS,i}(x)= I_{i-1}^{\sigma}(x)$. After this iteration, for every unnumbered neighbour $v$ of $\sigma(i)$, \\$l_{TBLS,i}(v)=l_{TBLS,i-1}(v)\cup\{i\}$, which is indeed $I_{i}^{\sigma}(v)$, and for every unnumbered non-neighbour $v$ of $\sigma(i)$, $l_{TBLS,i}(v)=l_{TBLS,i-1}(v)$, which is again $I_{i}^{\sigma}(v)$.

\end{proof}

\begin{theorem}\label{equivalence}
A set $T$ of orderings of the vertices of a graph $G$ is equal to $\{TBLS(G,\prec_{TBLS},\tau) \mid \tau\in \mathfrak{S}_n  \}$ if and only if there exists a labeling structure $S=(L,\prec_{GLS},l_0,UPLAB)$ such that $T$ is equal to the set of orderings produced by $GLS(G,S)$. 
\end{theorem} 
\begin{proof}
First, consider an ordering $\prec_{TBLS}$. The set $\{\mbox{TBLS}(G,\prec_{TBLS},\tau) \mid \tau\in \mathfrak{S}_n  \}$ is equal to the set of all orderings produced by GLS($G,S$) with $S=(\PN, \prec_{TBLS},\varnothing,cons)$, where $cons(l(w),i)$ returns $l(w)\cup\{i\}$.

Conversely, consider $S=(L,\prec_{GLS},l_0,UPLAB)$ a labeling structure. We show that there exists an order $\prec_{TBLS}$ such that, for every graph $G$, the set of all orderings produced by GLS($G,S$) is equal to $\{TBLS(G,\prec_{TBLS},\tau) \mid \tau\in \mathfrak{S}_n  \}$.

By propositions \ref{GLSlabelOK} and \ref{TBLSlabelOK} we can define  a mapping $\phi$ from $\PN$ (the labels used by TBLS) into labels  \emph{effectively} used by $GLS$ (i.e., those that can be assigned to a vertex during some execution of the algorithm). 
$\phi$ is recursively defined as $\phi(\varnothing)=l_0$, and if $max(A)=i$, then $\phi(A)=$ UPLAB($\phi(A) \backslash \{i\},i$).
Notice the same GLS-label $l$ may be reached in different ways. 
Subset $\phi^{-1} (l) \subset \PN$ is the set of TBLS-labels that correspond to that label. It is empty for all labels not effectively used.

Then, we define $\prec_{TBLS}$ as follows: $\forall A, A' \in \PNF$, $A\prec_{TBLS} A'$ if and only if $\phi(A)\prec_{GBLS} \phi(A')$.
We are now ready to prove the theorem.
The proof goes by induction. Before the first iteration, for every vertex $v$, $l_{GLS}^0(v)=l_0$ and $l_{TBLS}^0(v)=\varnothing$. $GLS$ can pick any of these vertices, in particular the one that would be picked by $TBLS (G,\prec_{TBLS},\tau)$, and by setting $\tau$ to be equal to a given output of $GLS(G,S)$, TBLS would indeed chose the same vertex.

Now, assume that when step $i$ begins, both algorithms have produced the ordering $\sigma(1) \dots \sigma(i-1)$, and the $i$th vertex is about to be chosen. By propositions  \ref{GLSlabelOK} and \ref{TBLSlabelOK}, for every unnumbered vertex $x$, $l_{TBLS,i}(x)= I_{i-1}^{\sigma}(x)$, and $l_{GLS,i}(x)= UPLAB (\dots UPLAB (l_o, p_{i-1}^1)\dots ,p_{i-1}^{|I_k^{\sigma}(x)|})$.
By the definition of $\phi$, we have that $l_{GLS,i}(x)=\phi(l_{TBLS,i}(x))$, and $l_{TBLS,i}(x)\in \phi^{-1}(l_{GLS,i}(x))$. Then, by the definition of $\prec_{TBLS}$, for two unnumbered vertices $v$ and $w$, we know that $l_{TBLS,i}(v) \prec_{TBLS} l_{TBLS,i}(w)$ if and only if $\phi(l_{TBLS,i}(v)) \prec_{GLS} \phi(l_{TBLS,i}(w))$, and $l_{GLS,i}(v) \prec_{GLS} l_{GLS,i}(w)$ if and only if for all $l_v\in \phi^{-1}(l_{GLS,i}(v))$ and all $l_w\in \phi^{-1}(l_{GLS,i}(w))$, $l_v\prec_{TBLS} l_w$. Thus, the set of eligible vertices at step $i$ is the same for both algorithms. $GLS$ can pick any of these vertices, in particular the one that would be picked by $\mbox{TBLS}(G,\prec_{TBLS},\tau)$, and by setting $\tau$ to be equal to GLS$(G,S)$, TBLS would indeed choose the right vertex.
\end{proof}

Although TBLS and GLS  cover the same set  of vertex orderings, we think that our TBLS formalism provides a simpler framework to analyze graph search algorithms, as can be seen in the next section.


\section{Recognition of some TBLS  search orderings}\label{tci}

Let us now  consider the following problem:

\vspace{0.5cm}

\textbf{Recognition of Search $\cal S$}

\KwData{Given a total ordering $\sigma$ of the vertices of a graph $G$  and a TBLS search $\cal S$,}

\KwResult{Does there exist $\tau$ such that $\sigma = \mbox{TBLS}(G, \prec_{\cal S}, \tau)$?}

\vspace{0.5cm}

Of course we can use Theorem~\ref{thm:fixpoint}
and build an algorithm that tests whether or not $\sigma=\mbox{TBLS}(G,\prec_{\cal S}, \sigma)$.  Let $\tau=\mbox{TBLS}(G,\prec_{\cal S},\sigma)$. 
If  $\tau=\sigma$ then the answer is yes; otherwise it is no.
We can certify the no answer using the first difference between $\tau$ and $\sigma$. Let $i$ be the first index such that $\sigma(i)\neq \tau(i)$. If TBLS chooses $\tau(i)$ and not $\sigma(i)$ at step $i$, then at this time $l(\sigma(i))\prec_{\cal S} l(\tau(i))$. So we can build a contradiction to the pattern-condition of this search.

But we may want to be able to answer this question without applying a TBLS search, or modifying a TBLS algorithm.
For example suppose that a distributed or parallel algorithm has been used to compute the ordering (for example when dealing with a huge graph 
\cite{BV12}) that is assumed to be a specific search ordering; how does one 
efficiently answer this question?  Let us study some cases.


\subsection{Generic Search}

For Generic Search consider Algorithm \ref{alg:gencertif} where $\sigma$ is the ordering we want to check, and for all $i$ between $1$ and $n$, $ln(\sigma(i))$ has
been computed; note that $G$ may be disconnected.  Recall that $ln(x)$ is the leftmost left neighbour of $x$; if $x$ has no left neighbours, then $ln(x) = -1$.  The algorithm will output either ``YES" or ``NO" depending on whether or not $\sigma$ is a GEN-ordering.

\vspace{0.5cm}

\begin{algorithm}[H] \label{alg:gencertif}
\caption{GEN-check}
$J\leftarrow 1$;  \hspace{4mm}\% \{ If $\sigma$ is a GEN-ordering, then $J$ is the index of the first \\ \hspace{1.5cm}vertex of the current connected component.\}\%\\
\For{$i\leftarrow 2$ \KwTo $n$}{
  \If{$ln(\sigma(i)) = -1$}{$J \leftarrow i$} \Else{{\If{$ln (\sigma(i)) < J$}{\Return{``NO''}}}
  }
  }
\Return{``YES''}

\end{algorithm}

\vspace{0.5cm}



\begin{theorem}\label{genericreco}
The GEN-check algorithm is correct and requires $O(n)$ time.  The recognition of a GEN-ordering can be implemented to run in $O(n+m)$  time.
\end{theorem}
\begin{proof}

If the algorithm reports that $\sigma$ is not a GEN-ordering, then vertices $\sigma (ln(i)),$$ \sigma (J), \sigma (i)$ form a forbidden triple as stipulated in Condition 2 of Theorem \ref{genericsearch}.  Note that $\sigma (J)$ has no neighbours to its left in $\sigma$.

Now assume that the algorithm reports that $\sigma$ is a GEN-ordering but for sake of contradiction there exists a forbidden triple on vertices $a <_{\sigma} b <_{\sigma} c$.
Let $J$ be the rightmost $J$ index less than $\sigma^{-1} (c) $ identified by the algorithm;  note that $ b \le_{\sigma} \sigma (J)  <_{\sigma} c$ and $ \sigma (ln(c)) \le_{\sigma} a$.  When $i = \sigma^{-1} (c)$ the algorithm would have reported that $\sigma$ is not a GEN-ordering.

For the preprocessing we need to compute the values of $ln (x)$ for every vertex $x$, following Definition \ref{lnrndef}.  By sorting the adjacency lists with respect to $\sigma$ (in linear time), it is possible to find $ln (x)$ in linear time by scanning the adjacency lists once and storing $ln (x)$ in an array.  Given this information, Algorithm 3 runs in $O(n)$ time.  Including the preprocessing time, the whole complexity needed is $O(n+m)$.
\end{proof}

\subsection{BFS}

In order to handle the recognition of BFS-orderings and DFS-orderings, we will first prove  variations of the conditions proposed in Theorems \ref{th:bfs} and \ref{th:dfs}, which are easier to check.
Let us  define for every vertex $x$ in V, the following two intervals in $\sigma$: $Right(x)=[x,rn(x)]$ and $Left(x)=[ln(x), x]$. By convention, if $rn(x)=-1$ or $ln(x)=-1$ the corresponding interval is reduced to $[x]$.

\begin{theorem}\label{BFScond}
Vertex ordering $\sigma$ is a BFS-ordering of $V$ if and only if \begin{enumerate}
\item Vertex ordering $\sigma$ is a GEN-ordering of $G$
\item For every pair of vertices $x, y$, if $x<_{\sigma} y$ then $ln(x) \leq_{\sigma} ln(y)$
\item For every pair of vertices $x, y$, if $x \ne y$ then the intervals $Left(x)$ and $Left(y)$ cannot be strictly included \footnote{One is included
in the other and the two left extremities are different, as are the two right extremities.}.
\end{enumerate}
\end{theorem}

\begin{proof}
It is easy to show that Conditions 2 and 3 are equivalent.

$\Rightarrow$ First, notice that every BFS-ordering $\sigma$ is also a GEN-ordering. 
Now assume for contradiction that Condition 3 is contradicted, namely that $x <_{\sigma} y$ and  that $Left(y)$ strictly contains $Left(x)$.  Then we have the configuration:  
$ln(y)<_{\sigma}ln(x) \le_{\sigma} x <_{\sigma} y$.
Considering the triple $(ln(y), x, y)$, since $ln(y)<_{\sigma} ln(x)$, necessarily $x ln(y) \notin E$.  Using the BFS  4-points condition 
on this triple there exists $z$ such that $z <_{\sigma} ln(y)$ where $xz \in E$, thereby contradicting $ln(y) <_{\sigma} ln(x)$. 

$\Leftarrow$
Assume that $\sigma$ respects all three conditions of the theorem.
Consider a triple $(a, b, c)$ of vertices such that:  $a<_{\sigma} b <_{\sigma} c$ with $ac \in  E$ and $ab \notin E$.
Since $\sigma$ is a GEN-ordering, $ac \in E$ implies that $ln(b) \neq -1$ (i.e., $b$ has a left neighbour in $\sigma$).

Suppose $ln(b) >_{\sigma} a$. Since $ln($$c)\leq_{\sigma}a$, this implies  that $Left(c)$ strictly contains $Left(b)$, thereby contradicting 
Condition 3.   Therefore $b$ has a neighbour before $a$ in $\sigma$. So $\sigma$ follows the BFS  4-points condition and is a legitimate BFS-ordering.

\end{proof}

To determine whether a given vertex ordering $\sigma$ is a BFS-ordering we first use Algorithm \ref{alg:gencertif} to ensure that $\sigma$ is a
GEN-ordering.  We then use Algorithm \ref{alg:BFScertif} to determine whether or not Condition 3 of Theorem \ref{BFScond} is satisfied and thus
whether or not $\sigma$ is a BFS-ordering.  As with Algorithm  \ref{alg:BFScertif}, we assume that $ln(\sigma(i))$ has been computed 
for all $i$ between $1$ and $n$.

\vspace{0.5cm}

\begin{algorithm}[H] \label{alg:BFScertif}
\caption{BFS-check}

$min \leftarrow n$;   \hspace{4mm}\%\{$min$ will store the index of the current leftmost value of \\ \hspace{1.7cm}$ln (\sigma (j))$ for all $i \le j \le n$.\}\% \\
	\For{i$\leftarrow$ n  downto 1}{
		\If{ $ln(\sigma(i)) > min $}{{\Return{``NO''}}}
		\If{$ln(\sigma(i)) \neq  -1$}{$min \leftarrow ln(\sigma(i))$;}
	}
	\Return{``YES''}
\end{algorithm}

\vspace{0.5cm}

\begin{theorem}\label{BFSreco}
Given a GEN-ordering $\sigma$, the BFS-check algorithm correctly determines whether $\sigma$ is a BFS-ordering
in $O(n)$ time.  The recognition of a BFS-ordering can be done in $O(n+m)$ time.
\end{theorem}
\begin{proof}

If the algorithm reports that $\sigma$ is not a BFS-ordering, then consider the triple of vertices $\sigma(min), \sigma(i), \sigma(k)$, where $k$ is the 
value of $i$ when $min$ was determined.  Note that $\sigma(i)$ is not adjacent to $\sigma(min)$ or to any vertices to the left of  $\sigma(min)$ and
thus this triple forms 
 a forbidden triple as stipulated in Condition 2 of Theorem \ref{th:bfs}.  

Now assume that the algorithm reports that $\sigma$ is a BFS-ordering but for sake of contradiction there exists a forbidden triple on vertices $a <_{\sigma} b <_{\sigma} c$.  We let $a' =\sigma (ln ($$c))$ and note that since $b$ has no neighbours to the left of or equal to $a$, $b$ is not adjacent to $a'$ or to any vertices to
its left.

Thus when $i = \sigma^{-1} (c)$ the algorithm would have reported that $\sigma$ is not a BFS-ordering.
The complexity argument is the same as in the proof of Theorem \ref{genericreco}.
\end{proof}

Concerning this particular result on  BFS, when the graph is connected it provides as a corollary a linear time algorithm to certify a shortest path between the vertices $\sigma (1)$ and $\sigma(n)$. So in the spirit of \cite{McMNS11}, this can be used for certifying  BFS-based diameter algorithms (see \cite{BV12, CGHLM13}).

\vspace{0.5cm}

\noindent

\subsection{DFS}

We now consider DFS and  define 
$Lmax(x)$ for every vertex $x \in V$ to be the rightmost left neighbour of $x$ in $\sigma$; if $x$ has no left neighbours then by convention $Lmax(x)=-1$.
The interval $RLeft(x)$ is defined to be 
$[Lmax(x), x]$; again by convention, if  $Lmax(x)=-1$ $RLeft(x)$  is reduced to $[x]$.

\begin{theorem}
Let $G=(V,E)$ be a graph, and let $\sigma$  be an ordering of $V$.  Vertex ordering 
$\sigma$ is a DFS-ordering of $G$ if and only if \begin{enumerate}
\item $\sigma$ is a GEN-ordering of $G$
\item no two intervals $Right(x)$ and $RLeft(y)$, with $x \neq y$, strictly  overlap as intervals.
\end{enumerate}
\end{theorem}

\begin{proof}
$\Rightarrow$ First, notice that every DFS-ordering is also a GEN-ordering. Then, assume, for contradiction, that $\sigma$ is a DFS-ordering of $G$, but that  in $\sigma$ $Right(x)$ and $RLeft(y)$  overlap for some $x \neq y$.  Necessarily $x <_{\sigma} y$
and $Lmax(y)<_{\sigma} x <_{\sigma} y <_{\sigma} rn(x)$.
$Lmax(y)<_{\sigma} x $ implies $xy \notin E(G)$.
 But then the triple $(x, y, rn(x))$ violates the 4-points condition
of $\sigma$, since $y$ has no neighbour  between $x$ and $y$ in $\sigma$.

$\Leftarrow$
Assume that $\sigma$ respects both conditions of the theorem but $\sigma$ is not a DFS-ordering.
Consider a triple $(a, b, c)$ of vertices such that:  $a<_{\sigma} b <_{\sigma} c$ with $ac \in  E$ and $ab \notin E$ but there is no neighbour
of $b$ 
 between $a$ and $b$ in $\sigma$.
Since $\sigma$ is supposed to  be a GEN-ordering, $ac \in E$ implies that $b$ has a neighbour $d$ left to it in $\sigma$, which by the 
above argument, must be before $a$.
Thus $Lmax(b) <_{\sigma} a$ and therefore  the intervals $RLeft(b), Right(a)$ strictly overlap, a contradiction.

\end{proof}

\begin{corollary}
DFS-orderings can be recognized  in $O(n+m)$.
\end{corollary}
\begin{proof}
Verifying that $\sigma$ is a generic-ordering can be done in $O(n+m)$ time using Theorem \ref{genericreco}.
To check the second condition, it suffices to build the  family of  $2n$ intervals and apply a simple 2 states stack automaton \cite{HMU01} to check the overlapping in $O(n)$ time.

\end{proof}

\subsection{LBFS  and LDFS}

\begin{theorem}\label{cert}
\LexBFS and \LexDFS-orderings can be recognized in $O(n(n+m))$ time.
\end{theorem}
\begin{proof} 

	To build the recognition algorithm we  use the third condition of the  relevant theorems in Section \ref{charac}, in particular
 \ref{th:lbfs} (LBFS) and \ref{th:ldfs} (LDFS).  Both of these conditions are pattern-conditions.
The certificate is stored in a table whose entries are keyed by the pair $(b,c)$ where $b <_{\sigma} c$ 
and the information will either be the vertex $a$, where $a  <_{\sigma} b$  that satisfies the corresponding condition or an error message
indicating that the condition has been violated. 
For LBFS and LDFS, the pattern-condition examines $a$, the leftmost (LBFS) or the rightmost (LDFS) vertex of $N(b) \bigtriangleup N(c) $ and requires that $a \in N(b)-N(c)$.  It is easy to show 
that this can be accomplished in time $O(|N(b)| + |N(c)|$, for any $b$ and $c$.
In all cases, if $a$ satisfies the membership condition then it is stored in the $(b,c)$'th entry of the table; otherwise ``error'' is stored.
	
Regarding complexity considerations, the table uses O($n^2$) space complexity.
For the lexicographic searches,
the timing requirement is bounded by $\sum_{b \in V} \sum_{c \in V} (|N(b)| + |N(c)|)$ to build the table and $O(n^2)$ time to
search for an ``error'' entry, giving an $O(n(n+m))$ time complexity.
 \end{proof}

These results for LBFS and LDFS do not seem to be optimal, but at least they yield a certificate in case of failure. To improve these algorithms we need to find some new characterizations of LBFS- and LDFS-orderings.

\section{Implementation  issues}\label{Imp}

We now  consider how to compute an TBLS search, in the case where $\prec$ is a total order. In such a case, at each step of the search, the labels yield a total preorder on the vertices. Such a  total preorder (also called weak-order using ordered sets  terminology) can be efficiently represented using ordered partitions as can beset in the next result.

\begin{theorem}\label{thpartoche}
$TBLS(G,\prec,\tau)$  where $\prec$ is a total order can be implemented to run
in $O(n+ mT(n)\log n)$ time where the  $\prec$ comparison time between two labels is bounded by $O(T(n))$.
\end{theorem}

\begin{proof}
We use the  framework of partition refinement \cite{HMPV00}. First we sort the adjacency lists with respect to $\tau$, and 
consider  the following algorithm.   
The input to the algorithm is a graph $G=(V,E)$, a total  order $\prec$ on $\PNF$,  and an ordering $\tau$ of $V$
and the output is the $\mbox{TBLS}(G,\prec,\tau)$-ordering $\sigma$ of $V$.

\vspace{0.5cm}

{\small
\begin{algorithm}[H] 
\caption{Computing a TBLS ordering}

Let $\mP$ be the  partition $\{V\}$, where the only  part (i.e., $V$) is ordered  with respect to $\tau$\;
\For{$i\leftarrow 1$ \KwTo $n$}
{
  Let $Eligible$ be the part of $\mP$ with the largest label with respect to $\prec$ and\; 
  let $x$ be its first vertex\;
  replace $Eligible$ by  $Eligible-\{x\}$ in $\mP$ \;
  $\sigma(i)\gets x$\;
  Refine($\mP$, $N(x)$)\;
}
\end{algorithm}
}

The algorithm maintains an (unordered) partition $\mP$ of the unnumbered vertices. Each part of $\mP$ is an ordered list of vertices. The following two invariants hold throughout the execution of the algorithm:

\begin{enumerate}
\item The vertices of each part have the same unique (with respect to parts) label;
\item Inside a part, the vertices are sorted with respect to $\tau$.
 \end{enumerate}

The action of Refine$(\mP,A)$ is to replace each part  $P\in\mP$ with two new parts: $P\cap A$ and $P-A$ (ignoring empty new parts).  It is possible to maintain the two invariants using the data structure from \cite{HMPV00}, provided the adjacency lists of $G$ are sorted with respect to $\tau$. After each refinement, each part of $\mP$ therefore contains vertices that are twins with respect to the visited vertices (Invariant 1). Thanks to the second invariant, the chosen vertex is always the first vertex (with respect to $\tau$) of part $Eligible$;  i.e., $\sigma(i)$ is indeed $x$. 

\vspace{0.5cm}

For the time complexity, Refine($\mP$, $N(x)$) takes $O(|N(x)|)$ time \cite{HMPV00}, so all refinements take $O(n+m)$ time. The only non-linear step is identifying part $Eligible$ among all parts of the partition. Each part has a label (the one shared with all its vertices) used as a key in a Max-Heap. Refine($\mP$, $N(x)$) creates at most $|N(x)|$ new parts so there are at most $m$ insertions into the heap. The label of a part does not change over time (but empty parts must be removed). There are no more removal operations than insertion operations, each consisting of at most $\log n$ label comparisons (since there are at most $n$ parts at any time). So we get the  $O(n+mT(n) \log n)$ time bound. 

\end{proof}

This complexity is not optimal, since
it is well-known and already used in some applications (see for example \cite{DOS09}) that classical searches such as BFS, DFS, LBFS  can be implemented  within the TBLS framework, i.e., solving the tie-break with a given total ordering $\tau$ of the vertices, within the same complexity as their current best implementations. To avoid the $T(n)$ costs and the $\log n$ factor, the trick  is simply to use an implementation of the search that uses partition refinement (such an implementation exists for BFS, DFS, and LBFS). If we start with a set ordered via $\tau$, there exists a partition refinement implementation that preserves this ordering on each part of the partition, and the tie-break rule means simply choose the first element of the Eligible part. For LDFS, that best known complexity can also be achieved this way. But for Gen-search, MCS and MNS we do  not know yet how to achieve linear time, within the TBLS framework.

\section{Concluding remarks}

We have focused our study on a new formalism that captures many usual graph searches as well as the commonly used
multi-sweep application of these searches.   The TBLS formalism  allows us to define a generic TBLS-orderings recognition algorithm, and gives us a new point of view of the hierarchy amongst graph searches. The new pattern-conditions for Generic Search, BFS and DFS give us a better way (compared to the pattern-conditions presented in \cite{CK}) of certifying whether a given ordering could have been produced by such a search. Furthermore, for LBFS and LDFS we do not have to trust the implementation of the search (which can be complicated) but have presented a simple program that just visits the neighbourhood of the vertices of the graph and stores a small amount of information (see Theorem \ref{cert}). The size of this extra information, however, can be bigger than the size of the input, and it may take longer to compute than the actual time needed to perform the search itself.

\vspace{0.3cm}

The landscape of graph search is quite complex. Graph searches can be clustered using the data structures involved in their best implementations (queue, stack, partition refinement \dots). In this paper we have tried a more formal way to classify graph searches. This attempt yields an algebraic framework that could be of some interest.

Clearly being an extension (see section \ref{sectTBLS})
is a transitive relation. In fact $\ll$ structures the TBLS graph searches as $\wedge$-semilattice.
The $0$ search in this semi-lattice, denoted by the null search or $S_{null}$, corresponds to the empty ordering relation (no comparable pairs).
At every step of $S_{null}$ the Eligible set contains all unnumbered vertices.
Therefore  for every $\tau$, $\mbox{TBLS}(G, \prec_{S_{null}}, \tau)=\tau$ and so any total ordering of the vertices can be produced by $S_{null}$.

The infimum between two searches $S,S'$ can be defined as follows:

\medskip
\noindent
For every pair of label sets $A, B$, we define: $A\prec_{S\wedge S'}B$   if and only if $A \prec_S B$ and $A \prec_{S'}B$.
\medskip

Clearly every extension of $S$ and $S'$ is an extension of $S\wedge S'$. Similarly $S$ and $S'$ are extensions of $S\wedge S'$.

While being as general as GLS, we feel that TBLS is closer to the pattern-conditions presented in \cite{CK}, since many of the $\prec$ conditions
presented in this paper are a rewriting of their pattern-conditions.
Still, there are many variants of the searches we studied that do not fall under the TBLS model, such as layered search. 
 We wonder if a more general search model can be found, that would not only include some of these other common searches but would also retain the simplicity of TBLS. 
  
 \section*{Acknowledgements:} The authors thank Dominique Fortin for his careful reading and numerous suggestions to improve the paper. 
The first author wishes to thank the Natural Sciences and Engineering Research Council of Canada for their financial support.


\end{document}